%% file: root.tex
\documentclass[letterpaper, 10 pt, conference]{ieeeconf}
\IEEEoverridecommandlockouts

\usepackage{amsmath, amssymb, amsfonts}
\usepackage{amsthm}
\usepackage{graphicx}
\usepackage{booktabs}
\usepackage{multirow}
\usepackage{makecell}
\usepackage{bm}
\usepackage{xcolor}
\usepackage{hyperref}
\usepackage{algorithm}
\usepackage{algpseudocode}
\let\labelindent\relax

\usepackage{enumitem}

\input{Macro}

\title{\LARGE \bf
Characterizing Identifiability and Generalization for Inverse Receding-Horizon Linear-Quadratic Regulator Problems
}

\author{Zhiyuan Jin$^{1}$, Jingqi Li$^{2}$ and David Fridovich-Keil$^{3}$
\thanks{$^{1}$Zhiyuan Jin is with the Department of Mathematics, University of Wisconsin–Madison, Madison, WI 53706, USA. {\tt\small zjin277@wisc.edu}}%
\thanks{$^{2}$Jingqi Li is with the Oden Institute for Computational Engineering and Sciences, The University of Texas at Austin, Austin, TX 78712, USA.{\tt\small  jingqi.li@austin.utexas.edu}}
\thanks{$^{3}$David Fridovich-Keil is with the Oden Institute for Computational Engineering and
Sciences, and the Department of Aerospace Engineering and Engineering Mechanics, The University of Texas at Austin, Austin, TX 78712, USA. {\tt\small dfk@utexas.edu}}
\thanks{This research was sponsored by the Army Research Laboratory under Cooperative Agreement W911NF-25-2-0021, and by the National Science Foundation under Grant 2336840. Jingqi Li is supported by the Peter O'Donnell Jr. Postdoctoral Fellowship. }
}

\begin{document}

\maketitle
\thispagestyle{empty}



\input{sections/abstract}
\input{sections/reviesion2_Introduction}
\input{sections/RelatedWorks}
\input{sections/revision2_ProblemFormulation}
\input{sections/revision2_Identifiability}
\input{sections/ObservationNumbers}
\input{sections/revision2_Generalization}

\input{sections/revision1_NumericalResults}

\input{sections/conclusion}
\bibliographystyle{IEEEtran}
\bibliography{references}
\end{document}

%% file: Macro.tex
\newcommand{\R}{\mathbb{R}}

\newcommand{\trans}{\mathsf{T}}
\newcommand{\Y}{\mathcal Y}

\newcommand{\rhcaction}[2]{u^0\!\left(#1;#2\right)}
\newcommand{\affhull}[1]{\operatorname{aff}\!\left\{#1\right\}}

\newtheorem{proposition}{Proposition}

\theoremstyle{definition}
\newtheorem{definition}{Definition}
\newtheorem{remark}{Remark}

\newtheorem{example}{Example}

\newcommand{\dataset}{\mathcal{Y}_K^\star}
\newcommand{\observationmap}{\mathcal{O}}
\newcommand{\parameterspace}{\Theta}
\newcommand{\parameter}{\theta}

%% file: sections/abstract.tex
\begin{abstract}
We consider the problem of objective inference in the context of receding-horizon linear-quadratic regulator (LQR).
In this setting, we are given sequential state-action observations, where each observed action is the first control of a newly solved finite-horizon LQR problem. 
We characterize when the objective of that problem is uniquely identifiable from these observations and when additional observations provide no new information about the objective. 
We then analyze action prediction at unseen states and show that all objectives reproducing the observed actions yield identical actions throughout the affine hull of the observed states; outside this hull, we derive an upper bound on the prediction error. 
Additionally, we show that, when only the linear objective terms are unknown, exact prediction holds at every state. 
Finally, numerical results show that, even under stochastic observation noise, re-optimizing an inferred objective enables accurate action prediction at unseen states across different planning horizons.
\end{abstract}


%% file: sections/reviesion2_Introduction.tex
\section{Introduction}
Receding-horizon control (RHC) is widely used for sequential decision-making because repeated replanning allows a controller to adapt to evolving states and environments \cite{mattingley2011receding}. Inferring the objective underlying these decisions provides an interpretable description of the preferences and tradeoffs that shape the controller’s behavior \cite{zhang2024inverse, 11312871}. Beyond interpretation, the inferred objective can be re-optimized to predict actions at new states or under different planning horizons.

Inverse optimal control (IOC) provides a general framework for recovering such objectives from observed decisions \cite{Azar2020FromIO}. For theoretical analysis, the linear-quadratic regulator (LQR) setting is particularly tractable because the dependence of optimal control on the objective parameters can be characterized explicitly. Existing inverse LQR results establish conditions for objective recovery from feedback laws~\cite{Kalman1964} or complete finite-horizon optimal trajectories \cite{ZhangEtAl2019}.

Receding-horizon control, however, produces a fundamentally different form of data. At each state, a finite-horizon problem is solved, but only the first action is applied before replanning. Consequently, sequential observations are drawn from a sequence of distinct finite-horizon optimization problems rather than from a single optimal trajectory. Existing inverse receding-horizon methods \cite{zhang2024inverse,11312871} can recover objectives that reproduce the observed first actions, but they do not reveal whether the underlying objective is uniquely identifiable. Nor do they determine whether inferred objectives can predict ground-truth actions at unseen states.

We study these questions in the inverse receding-horizon LQR setting. Our main results are as follows.

 1) \textbf{Identifiability from observations.}
We characterize when the objective is uniquely identifiable from observed first actions and derive necessary dimensional conditions for unique recovery. When only the linear objective terms are unknown, we derive identifiability conditions and characterize the identifiable parameter subspace.

 2) \textbf{Information gain from additional observations.}
We characterize when additional observations provide no new information about the objective. When quadratic objective terms are unknown, observations within the affine hull of previously observed states add no information. When only the linear objective terms are unknown, a single exact observation contains all available objective information, so subsequent observations do not improve inference.

 3) \textbf{Generalization to unseen states.}
We show that all objectives reproducing the observed first actions induce the same first action throughout the affine hull of the observed states. Outside this hull, we derive a bound on the prediction error at unseen states. When only the linear objective terms are unknown, exact prediction holds at every state. Numerically, re-optimizing the inferred objective yields accurate unseen-state prediction across different planning horizons, even under stochastic observation noise.

%% file: sections/RelatedWorks.tex
\section{Related Work}
\noindent\textbf{Inverse Optimal Control.} IOC problems aim to infer objective function parameters that explain observed controls~\cite{Kalman1964}.
The literature on inverse LQR specifically studies objective identifiability from stationary feedback policies~\cite{PriessEtAl2015,ZhangRingh2024} or finite-horizon trajectories~\cite{ZhangEtAl2019,QuHeDuan2025}.
However, in practice, receding-horizon control reveal only the first action of each replanned trajectory.
In contrast to prior works on inverse RHC~\cite{RamadanEtAl2019,XuGaoHsu2022}, we focus on the LQR case, whose structure allows us to characterize the set of recoverable objective parameters, the observations that provide new information about them, and the action-prediction error at unseen states.

\noindent\textbf{Data-Driven Control.} Data-driven control constructs controllers from trajectory data through system identification and synthesis~\cite{Simchowitz2018,Dean2020,Mania2019}, direct feedback design~\cite{DePersisTesi2020}, or data-enabled predictive control~\cite{Coulson2019}, while data-informativity analyses distinguish data sufficient for control from that required for unique system identification~\cite{VanWaarde2020}. We complement these works by retaining the optimization structure and inferring the objective underlying receding-horizon LQR behavior. From observed first actions, we characterize objective identifiability, when additional observations provide new information, and when inferred objectives correctly predict actions at unseen states.

\noindent\textbf{Inverse Reinforcement Learning.} 
Inverse reinforcement learning (IRL) infers objectives from demonstrations~\cite{NgRussell2000,AbbeelNg2004}, with probabilistic formulations for imperfect behavior~\cite{Ziebart2008} and expressive neural-network representations of objective models for complex tasks~\cite{Finn2016}. While these methods address increasingly complex settings, exact analysis of objective identifiability, of the information gained from additional observations, and of action-prediction generalization at unseen states remains difficult. We study a tractable inverse receding-horizon LQR setting that permits explicit characterization of all three problems. 

%% file: sections/revision2_ProblemFormulation.tex
\section{Preliminaries and Problem Formulation}
\label{sec:problem_formulation}

\subsection{System Setup and Receding-Horizon Control}
Consider the discrete-time linear system
\begin{equation}
    x_{t+1}=Ax_t+Bu_t,
    \label{eq:system_dynamics}
\end{equation}
where $x_t\in\R^n$ is the state, $u_t\in\R^m$ is the control input,
and $A\in\R^{n\times n}$ and $B\in\R^{n\times m}$ are known and fixed.

The stage cost is
$    \ell(x_{t+1},u_t;\parameter)
    ={}\frac12 x_{t+1}^\trans Q(\parameter)x_{t+1}
    +\frac12 u_t^\trans R(\parameter)u_t
    +q(\parameter)^\trans x_{t+1}
    +r(\parameter)^\trans u_t$,
where $Q(\parameter)\in\R^{n\times n}$ and $R(\parameter)\in\R^{m\times m}$
are symmetric, $q(\parameter)\in\R^n$, and $r(\parameter)\in\R^m$. Since a common positive
scaling of all coefficients leaves the optimal control unchanged, we remove
this ambiguity by imposing $\operatorname{tr}\big(R(\parameter)\big)=1$.
The parameter $\parameter\in\R^p$ then collects the $p$
independent unknown coefficients, and each $\parameter$ uniquely specifies
the cost tuple $\big(Q(\parameter),R(\parameter),q(\parameter),r(\parameter)\big)$.
The feasible parameter set is
$\Theta:=\{\parameter\in\R^p: Q(\parameter)\succeq0,\ R(\parameter)\succ0\}$.



For a fixed and known planning horizon $T\ge1$ and current state $x$, consider the
finite-horizon optimal control problem
\begin{equation}
\begin{aligned}
    \min_{u_0,\ldots,u_{T-1}}\quad
        &\sum_{\tau=0}^{T-1}\ell(x_{\tau+1},u_\tau;\theta)\\
    \text{subject to}\quad
        &x_0=x,\\
        &x_{\tau+1}=Ax_\tau+Bu_\tau, \ \tau=0,\ldots,T-1.
\end{aligned}
    \label{eq:forward_rhc_problem}
\end{equation}
Here $\tau$ indexes a predicted trajectory within a planning problem,
whereas $t$ indexes the actual closed-loop evolution. Since $R\succ 0$, the optimal solution to \eqref{eq:forward_rhc_problem} is unique. Let
$U^\star(x;\theta):=[u_0^{\star\trans},\ldots,u_{T-1}^{\star\trans}]^\trans$ denote
the unique optimal control sequence, and define its first action as
$\rhcaction{x}{\theta}:=u_0^\star$.

RHC solves
\eqref{eq:forward_rhc_problem} at each time $t$ with initial state
$x_t$, applies only $\rhcaction{x_t}{\theta}$, and then solves a
new horizon-$T$ problem at the next state. This leads to a
closed-loop system
\begin{equation}
    x_{t+1}=Ax_t+B\rhcaction{x_t}{\theta}.
    \label{eq:rhc_closed_loop}
\end{equation}

\subsection{Problem Formulation: Inverse Receding-Horizon LQR}
The inverse problem seeks objective parameters whose corresponding optimal control matches the observed data. The data-generating controller, referred to as the expert, uses a fixed
parameter $\theta^\star\in\Theta$ and generates demonstrations by
repeatedly solving \eqref{eq:forward_rhc_problem} and applying its
first action:
\begin{equation}
\begin{aligned}
    u_t^{0\star} =\rhcaction{x_t}{\theta^\star},\  \ \ 
    x_{t+1} =Ax_t+Bu_t^{0\star}.
\end{aligned}
    \label{eq:expert_demonstration}
\end{equation}
The available data is
$\dataset=\{(x_t,u_t^{0\star})\}_{t=0}^{K-1}$, where $K\ge1$
is the number of observed state-action pairs. Only the applied
first actions are observed; the remaining actions in each expert
plan are unobserved.

Given $\Y_K^\star$, the inverse receding-horizon LQR problem is
\begin{equation}
    \min_{\hat\theta\in\Theta}\;
    \mathcal L_K(\hat\theta)
    :=\sum_{t=0}^{K-1}
    \|\rhcaction{x_t}{\hat\theta}-u_t^{0\star}\|_2^2.
    \label{eq:inverse_rhc_problem}
\end{equation}

Previous works have focused on computational methods for finding solutions minimizing the loss function \eqref{eq:inverse_rhc_problem} for inverse RHC problems \cite{zhang2024inverse, 11312871}. However, achieving zero loss guarantees only that the inferred parameter reproduces the observed actions at the demonstrated states; a priori, it neither ensures recovery of the ground-truth objective parameter nor provides any formal analysis of the error in predicting actions at unseen states. We therefore introduce \emph{identifiability} and \emph{generalizability} to characterize when the observations uniquely determine the true parameter and whether an inferred parameter can predict the expert's actions at unobserved states even when recovery is not unique.

First, different objective parameters may induce the same observed first actions. Identifiability concerns whether the observations uniquely determine the objective parameter.

\begin{definition}[Identifiability]
\label{identifiability}
Given a dataset $\dataset$, we define the observation map from objective parameters and observed states to observed controls $\observationmap: \parameterspace \times \mathbb{R}^{Kn} \to \mathbb{R}^{Km} $. A parameter $\theta^\star\in\Theta$ is \emph{identifiable} if $\{\theta\in\Theta: \mathcal O(\theta, \{x_t\}_{t=0}^{K-1})=\mathcal O(\theta^\star,\{x_t\}_{t=0}^{K-1})\} =\{\theta^\star\}$, equivalently, if $\observationmap$ is injective in~$\parameter$. 
\end{definition}

Moreover, accurate action prediction may still be possible even when the objective parameter is not uniquely identifiable. Generalizability concerns whether the inferred objective reproduces the expert's first action at an unseen query state under the same dynamics and planning horizon.
\begin{definition}[Generalizability]
Let $\theta^\star$ denote the ground-truth objective parameter.
A parameter $\theta$ is \emph{generalizable} at state $\bar x$ if $\rhcaction{\bar x}{\theta}=\rhcaction{\bar x}{\theta^\star}$.
\end{definition}
When generalizability cannot be guaranteed, we define the \emph{generalization error} as the action-prediction error at an unseen query state, i.e., $\|\rhcaction{\bar x}{\theta}-\rhcaction{\bar x}{\theta^\star}\|_2$.


\subsection{Lifted Optimality Representation}
The inverse problem in \eqref{eq:inverse_rhc_problem} expresses action matching through the solution of a finite-horizon optimization problem. To analyze which cost parameters can reproduce the observations, we first derive explicit optimality conditions for the planned control sequence. These conditions will allow the unobserved planned actions to be treated as latent variables. We then extract the first action and show
that it depends affinely on the current state, providing the structure used to analyze prediction at unseen states. 

Define the \emph{stacked state and
control trajectories} as
$
    X:=[x_1^\trans, \ldots,  x_T^\trans]^\trans,\ 
    U:=[u_0^\trans, \ldots,  u_{T-1}^\trans]^\trans
$. 
The dynamics admit the lifted representation $X=\mathcal A_Tx+\mathcal B_TU$, where 
\[
    \mathcal A_T=
    \begin{bmatrix}A\\A^2\\\vdots\\A^T\end{bmatrix}
    ,
    \mathcal B_T=
    \begin{bmatrix}
        B&0&\cdots&0\\
        AB&B&\cdots&0\\
        \vdots&\vdots&\ddots&\vdots\\
        A^{T-1}B&A^{T-2}B&\cdots&B
    \end{bmatrix}
    .
\]
Furthermore, define
$
    \bar Q:=I_T\otimes Q,\ 
    \bar R:=I_T\otimes R,\ 
    \bar q:=\mathbf 1_T\otimes q,\ 
    \bar r:=\mathbf 1_T\otimes r
$. 
After eliminating the state trajectory, the finite-horizon objective
becomes
\begin{align}
    J(U;x,\theta)
    ={}&
    \frac{1}{2}
    (\mathcal A_Tx+\mathcal B_TU)^\trans
    \bar Q
    (\mathcal A_Tx+\mathcal B_TU)
    \nonumber\\
    &+
    \bar q^\trans(\mathcal A_Tx+\mathcal B_TU)
    +\frac{1}{2}U^\trans\bar R U
    +\bar r^\trans U.
    \label{eq:lifted_objective}
\end{align}
The optimal control trajectory $U^\star(x;\theta)$ therefore satisfies
the first-order optimality condition
\begin{equation}
    \left(
        \bar R+\mathcal B_T^\trans\bar Q\mathcal B_T
    \right)U^\star
    +
    \mathcal B_T^\trans\bar Q\mathcal A_Tx
    +
    \mathcal B_T^\trans\bar q
    +
    \bar r
    =0.
    \label{eq:stacked_optimality}
\end{equation}

Since $Q\succeq 0$ and $R\succ 0$, the objective is strictly convex in $U$. Let $E_1=[I_m\;0]\in\R^{m\times Tm}$ extract the first action from a stacked control sequence. Then the first-action policy is
\begin{equation}
    \rhcaction{x}{\theta}=E_1U^\star(x;\theta)=M_\theta x+b_\theta,
    \label{eq:first_action_policy}
\end{equation}
where $M_\theta
    :=-E_1
    \left(\bar R+\mathcal B_T^\trans\bar Q\mathcal B_T\right)^{-1}
    \mathcal B_T^\trans\bar Q\mathcal A_T$ and $ b_\theta
    :=-E_1
    \left(\bar R+\mathcal B_T^\trans\bar Q\mathcal B_T\right)^{-1}
    \left(\mathcal B_T^\trans\bar q+\bar r\right)$.

%% file: sections/revision2_identifiability.tex
\section{Objective Identifiability from Observations}
\label{sec:identifiability}

Different objective parameters may generate the same receding-horizon behavior, making unique recovery impossible in some cases.
We derive a set of necessary and sufficient conditions for unique recovery, and characterize when additional observations provide new information. 

\subsection{Inferring General Quadratic Objective Terms}

\subsubsection{Necessary Conditions for Identifiability}

Before asking whether the objective can be uniquely recovered, we first examine whether the observations can contain enough information to distinguish its unknown parameters. The following proposition identifies two-dimensional limitations: one imposed by the number of observed actions and another imposed by the structure of the fixed-horizon policy itself.

\begin{proposition}[Dimensional necessary conditions]
\label{prop:dimensional_conditions}
Uniquely identifying a $p$-dimensional objective from $K$
first-action observations requires both $p \le Km$ and
$p \le m(n+1)$.
If $q=r=0$, the policy is linear and the latter condition becomes
$p\le mn$.
\end{proposition}

\begin{proof}
An injective continuous map cannot go from a $p$-dimensional space into a space of dimension less than $p$. 
Fixing the data set $\dataset$, the image of the observation map $\observationmap(\theta,\{x_t\}_{t=0}^{K-1})$ is in $\mathbb R^{Km}$, so injectivity in $\theta$ requires $p\le Km$. 
Moreover, since we observe receding-horizon control data, every observation obeys the same first-action policy $u^0(x;\theta)=M_\theta x+b_\theta$ of \eqref{eq:first_action_policy} evaluated at a different state. 
If the map from $\parameter$ to $ (M_\parameter, b_\parameter)$ is not injective, then $\observationmap(\theta,\{x_t\}_{t=0}^{K-1})$ loses injectivity in $\parameter$. 
Since the pair $(M_\theta,b_\theta)\in\mathbb R^{m\times n}\times\mathbb R^m$ has $m(n+1)$ entries, 
to ensure the injectivity of $\theta\mapsto(M_\theta,b_\theta)$, we require $p\le m(n+1)$. 
When $q=r=0$, $b_\theta=0$ \eqref{eq:first_action_policy} gives $b_\theta=0$, so the policy is determined by $M_\theta$ alone, and the bound becomes $p\le mn$.
\end{proof}

The following example illustrates the failure of 
identifiability when the dimensional conditions are violated. In this case, the Jacobian of the observed first-action
predictions is necessarily rank deficient.
\begin{example}[Failure of the dimensional conditions]
Consider $n=m=1$ and a normalized objective parameter of
dimension $p=3$. Since
$\rhcaction{x}{\theta}=M_\theta x+b_\theta$, the Jacobian
of the predictions in \eqref{eq:first_action_policy} satisfies
\[
\frac{\partial}{\partial\theta}
\begin{bmatrix}
\rhcaction{x_0}{\theta}\\
\vdots\\
\rhcaction{x_{K-1}}{\theta}
\end{bmatrix}
=
\begin{bmatrix}
x_0 & 1\\
\vdots & \vdots\\
x_{K-1} & 1
\end{bmatrix}
\frac{\partial(M_\theta,b_\theta)}{\partial\theta}.
\]
Its rank is at most $\min\{K,2\}<p$, regardless of the
observed states. With $K=1$, the observation-count condition
$p\le Km$ fails; for any $K$, the policy-dimension condition
$p\le m(n+1)$ fails.
\end{example}

\begin{remark}[Linear costs can improve identifiability]
When $q=r=0$, the first-action policy is linear and has at most $mn$ degrees of freedom. Allowing nonzero linear cost terms $q$ or $r$ generally introduces an affine offset to the RHC policy, increasing the policy dimension to at most $m(n+1)$. The additional $m$ degrees of freedom can make unique identification possible for a structured parameterization that could not be identified from a purely linear policy. 
\end{remark}

\subsubsection{Exact Characterization of Identifiability}

The dimensional conditions above determine when exact recovery is impossible,
but they do not provide a sufficient condition for identifiability. We now
give an exact characterization of the set of objective parameters that can reproduce all observed first actions.

The following construction is independent of which cost terms are unknown, i.e.,
the parameter $\theta$ may represent any structured subset of $Q$, $R$, $q$,
and $r$.

Using the lifted notation, define the optimality residual $\mathcal F(U,x;\theta)
    :=\left(\bar R+\mathcal B_T^\trans\bar Q\mathcal B_T\right)U
    +\mathcal B_T^\trans\bar Q\mathcal A_Tx
    +\mathcal B_T^\trans\bar q+\bar r$, 
where the cost matrices and vectors are those specified by $\theta$. Because
$Q\succeq0$ and $R\succ0$, a control sequence is optimal if and only if
$\mathcal F(U,x;\theta)=0$.

Only first action is observed. For the observation $(x_t,u_t^{0\star})$, write a possible completion as $U_t(V_t)=
    [    u_t^{0\star\trans},  V_t^\trans     ]^\trans,
    \ V_t\in\mathbb R^{(T-1)m}$. 
The parameters consistent with all observations are
\begin{equation}
\begin{split}
    \mathcal C_K(\mathcal Y^\star_K)
    :=\bigcap_{t=0}^{K-1}
    \biggl\{\theta\in\Theta:\;&\exists V_t\text{ such that}\\[-1mm]
    &\mathcal F\bigl(U_t(V_t),x_t;\theta\bigr)=0
    \biggr\}.
    \label{eq:consistency_set}
\end{split}
\end{equation}
Since the data set $\mathcal{Y}^\star_K$ is fixed throughout, we abbreviate
$\mathcal{C}_K := \mathcal{C}_K(\mathcal{Y}^\star_K)$ whenever the data set is
clear from context. For $j \le K$, $\mathcal{C}_j$ denotes the set formed from the
first $j$ observations $\{(x_t,u^{0\star}_t)\}_{t=0}^{j-1}$ of the same trajectory,
so that $\mathcal{C}_{K+1} \subseteq \mathcal{C}_K \subseteq \cdots \subseteq \mathcal{C}_1$.
By construction, \(\mathcal C_K(\mathcal Y^\star)\) contains the true parameter \(\theta^\star\) and every admissible alternative that reproduces the observations. Identifiability therefore reduces to whether the observations exclude all such alternatives.
\begin{proposition}[Consistency characterization]
\label{prop:consistency_identifiability}
The normalized parameter $\theta^\star$ is identifiable from
$\mathcal Y^\star$ if and only if $ \mathcal C_K$ is a singleton.
\end{proposition}
\begin{proof}
Since $Q\succeq0, R\succ 0$, the optimal sequence from $x_t$ under $\theta$ is unique and is the only solution of $\mathcal F(\,\cdot\,,x_t;\theta)=0$. Hence there exists a completion $U_t(V_t)$ satisfying $\mathcal F(U_t(V_t),x_t;\theta)=0$ if and only if $u^0(x_t;\theta)=u^{0\star}_t$. Intersecting over $t$ gives $\mathcal C_K=\{\theta\in\Theta:\mathcal O(\theta,\{x_t\}_{t=0}^{K-1})=\mathcal Y^\star\}$, and the claim follows from Definition~\ref{identifiability}.
\end{proof}
Proposition~\ref{prop:consistency_identifiability} turns identifiability into a question about the solution set of the consistency equations. Since each residual $\mathcal F(U_t(V_t),x_t;\theta)=0$ is bilinear in $\theta$ and $V_t$, these equations form a polynomial system in $(\theta,V_0,\ldots,V_{K-1})$. When $\mathcal C_K$ is finite, its members can therefore be enumerated numerically and the singleton test applied directly.\footnote{\label{fn:homotopy}We use \texttt{HomotopyContinuation.jl} \cite{BreidingTimme2018} to solve the polynomial systems. Among the returned solutions for $(\theta,V_0,\ldots,V_{K-1})$ we keep the $\theta$-components that are real, satisfy the normalization, and yield $Q(\parameter)\succeq0$ and $R(\parameter)\succ0$; each is an isolated member of $\mathcal C_K$.} When $\mathcal C_K$ is not finite, the consistency equations instead describe a continuum of indistinguishable objectives. The next subsection treats a case where this continuum can be characterized in closed form.

\subsection{Special Case: Inferring Only Linear Objective Terms}

Suppose $Q$ and $R$ are known and only $(q,r)$ is unknown. We write
$\rhcaction{x}{q,r}$ for the resulting first action. 
In this case, the gain $M_\theta$ in \eqref{eq:first_action_policy} is a fixed matrix $M$, while $b_\theta$ is linear in $(q,r)$. Hence
\begin{equation}
    \rhcaction{x}{q,r}
    =Mx+L
    \begin{bmatrix}q\\r\end{bmatrix},
    \label{eq:affine_cost_map}
\end{equation}
where
$    L=-E_1
    \left(\bar R+\mathcal B_T^\trans\bar Q\mathcal B_T\right)^{-1}
    \begin{bmatrix}
        \mathcal B_T^\trans(\mathbf1_T\otimes I_n)
        &\mathbf1_T\otimes I_m
    \end{bmatrix}.$

Thus $L$ is precisely the sensitivity of the common offset $b_\theta$ to the joint linear-cost vector $[q^\trans,r^\trans]^\trans$. We formalize the identifiability result in the following proposition. 

\begin{proposition}[Linear-cost and identifiable subspace]
\label{prop:affine_cost_identifiability}
For every $K\ge1$, $\mathcal C_K=[q^{\star\trans} ,\ r^{\star\trans}]^\trans +\ker(L)$. 
Consequently, no parameter is identifiable under unrestricted joint $(q,r)$ inference, because $L\in\mathbb R^{m\times(n+m)}$. If only $q$ or only $r$ is unknown, identifiability is determined by the column rank of the corresponding block of $L$. Moreover, the identifiable parameter subspace for joint $(q,r)$ inference is $\mathcal I_{qr}:=\ker(L)^\perp=\operatorname{im}(L^\trans)\subseteq \mathbb R^{n+m}$.
\end{proposition}

\begin{proof}
Let $z=[q^\trans,r^\trans]^\trans$ and $z^\star=[q^{\star\trans},r^{\star\trans}]^\trans$. By \eqref{eq:affine_cost_map}, a candidate $z$ matches an observed first action at any state $x_t$ if and only if $Mx_t+Lz=Mx_t+Lz^\star$, 
or equivalently $L(z-z^\star)=0$. This condition is independent of $x_t$, so every nonempty collection of observations yields
$\mathcal C_K=z^\star+\ker(L)$.
Since $L$ has $m$ rows and $n+m$ columns, rank-nullity gives $\dim\ker(L)\ge n>0$, precluding unique joint recovery. If only one block is unknown, the same argument applied to the corresponding block of $L$ shows that recovery is unique exactly when that block has full column rank.
\end{proof}





%% file: sections/ObservationNumbers.tex
\section{Information Gain from Additional Observations}

Longer trajectories are useful only while their observations continue to reduce parameter ambiguity. We characterize when a new observation is redundant, i.e., $\mathcal C_{K+1}=\mathcal C_K$. The answer depends on the geometry of the observed states and on which cost terms are unknown. For $K\ge1$, define the observed affine hull
\begin{multline}
    \affhull{x_0,\ldots,x_{K-1}}
    =\\x_0+\operatorname{span}\{x_1-x_0,\ldots,x_{K-1}-x_0\}.
    \label{eq:observed_affine_hull_span}
\end{multline}

\begin{proposition}[No information gain]
\label{prop:no_gain}
Let $\mathcal C_K$ and $\mathcal C_{K+1}$ be computed under $K$ and
$K+1$ first-action observations, respectively. Suppose one of the following conditions holds:
\begin{enumerate}[label=(\roman*)]
  \item $x_K\in\affhull{x_0,\ldots,x_{K-1}}$;
  \item if $q=r=0$ throughout~$\Theta$, $x_K\in\operatorname{span}\{x_0,\ldots,x_{K-1}\}$;
\end{enumerate}
Then we have $\mathcal C_{K+1}=\mathcal C_K$, i.e., an additional observation cannot refine the inferred objective parameter~set. Moreover, when only the linear objective term $(q,r)$ is unknown, $\mathcal C_K=\mathcal C_1$ for all $K\ge1$.
\end{proposition}
\begin{proof}
$\mathcal C_{K+1} \subseteq \mathcal C _K$ holds trivially; it suffices to show the reverse inclusion $\mathcal C_K\subseteq\mathcal C_{K+1}$, i.e., every $\theta$ consistent with the first $K$ observations is also consistent with $(x_K,u^{0\star}_K)$.

By \eqref{eq:consistency_set}, $\theta$ is consistent with the new observation
$(x_K,u^{0\star}_K)$ if and only if some control sequence starting with
$u^{0\star}_K$ is optimal under $\theta$ from $x_K$. Because the optimal
sequence is unique, this is the same as requiring its first entry to be
$u^{0\star}_K$, that is,$    u^0(x_K;\theta)=u^0(x_K;\theta^\star).$
Hence $\mathcal C_K\subseteq\mathcal C_{K+1}$ holds if agreement at the
observed states, $u^0(x_t;\theta)=u^0(x_t;\theta^\star)$ for $t<K$, implies
agreement at $x_K$. 

We verify this in each case,
using that every first-action policy has the affine form
$u^0(x;\theta)=M_\theta x+b_\theta$.

\emph{(i)} By \eqref{eq:observed_affine_hull_span}, $x_K\in\affhull{x_0,\ldots,x_{K-1}}$ means $x_K=x_0+\sum_{t=1}^{K-1}\mu_t(x_t-x_0)$ for some $\mu_1,\ldots,\mu_{K-1}\in\mathbb R$. 
Then $u^0(x_K;\theta)
= M_\theta x_0+b_\theta+\sum_{t=1}^{K-1}\mu_t M_\theta(x_t-x_0)= u^0(x_0;\theta)+\sum_{t=1}^{K-1}\mu_t\bigl(u^0(x_t;\theta)-u^0(x_0;\theta)\bigr).$
Since $\theta\in\mathcal C_K$, we have $u^0(x_t;\theta)=u^0(x_t;\theta^\star)$ for every $t<K$, and therefore $u^0(x_K;\theta)=u^0(x_K;\theta^\star)$.

\emph{(ii)} If $q=r=0$ throughout $\Theta$, then by \eqref{eq:first_action_policy} $b_\theta=0$ for every $\theta$, so $u^0(x;\theta)=M_\theta x$ is linear in $x$. Then for any $x_K=\sum_{t=0}^{K-1}\mu_t x_t\in\operatorname{span}\{x_0,\ldots,x_{K-1}\}$, with no restriction on the $\mu_t$, $u^0(x_K;\theta)=\sum_{t=0}^{K-1}\mu_t u^0(x_t;\theta)$, and likewise for $\theta^\star$. The argument of (i) applies verbatim: agreement at $x_0,\ldots,x_{K-1}$ gives $u^0(x_K;\theta)=u^0(x_K;\theta^\star)$.

If only $(q,r)$ is unknown, then by \eqref{eq:affine_cost_map} the gain is the same fixed matrix $M$ for every $\theta$, so $u^0(x;\theta)-u^0(x;\theta^\star)=b_\theta-b_{\theta^\star}$ is a constant independent of $x$. Since $\theta\in\mathcal C_K$ requires agreement at $x_0$, this constant is zero, and the policies agree at every state, in particular at $x_K$. As this holds for every $K\ge1$ and every $x_K$, induction gives $\mathcal C_K=\mathcal C_1$.
\end{proof}
\begin{remark}
\label{rem:no_gain_practice}
Proposition~\ref{prop:no_gain} has two practical consequences.
\emph{1)~Efficient inference:} an observation whose state lies in the
affine hull of the previous states (the span in case (ii), any state in case
(iii)) does not change $\mathcal C_K$, so it can be dropped and the
consistency equations \eqref{eq:consistency_set} formed only over the
remaining observations; in the linear-cost case one observation is enough.
\emph{2)~Data collection:} only a state outside the current affine hull can
shrink $\mathcal C_K$, so new observations should be taken at such states.
\end{remark}

Proposition~\ref{prop:no_gain} concerns what the data reveal about the
parameter. The next section asks the converse: what a consistent parameter,
identifiable or not, predicts about the expert at states never observed. The
same regions reappear, now as the states on which prediction is exact.

%% file: sections/revision2_Generalization.tex
\section{Generalization of Action Prediction to Unseen States}
\label{sec:generalization}

Identifiability concerns recovery of the objective parameters, whereas
generalization concerns recovery of the induced behavior: a parameter may fail
to be identifiable yet still produce the correct first action at unseen
states. 

We first characterize the states on which every consistent parameter is guaranteed to generalize and then derive a prediction-error bound for the remaining states.

\begin{proposition}[Exact generalization]
\label{prop:exact_generalization}
Let $\theta\in\mathcal C_K$ be any parameter consistent with
the past observations. Then the following hold.
\begin{enumerate}
  \item[(i)] For any parameterization $\Theta$, $\theta$ predicts the
        expert's first action exactly at every state in the affine hull of
        the observed states:
        $\rhcaction{x}{\theta}=\rhcaction{x}{\theta^\star}$ for all
        $x\in\affhull{x_0,\ldots,x_{K-1}}$.
  \item[(ii)] If $q=r=0$ throughout $\Theta$, the guarantee extends to the
        linear span of the observed states:
        $\rhcaction{x}{\theta}=\rhcaction{x}{\theta^\star}$ for all
        $x\in\operatorname{span}\{x_0,\ldots,x_{K-1}\}$.
  \item[(iii)] When only the linear objective term$(q,r)$ is unknown, the guarantee extends to the whole
        state space: $\rhcaction{x}{\theta}=\rhcaction{x}{\theta^\star}$ for
        all $x\in\R^n$. 
\end{enumerate}
\end{proposition}
\begin{proof}
The proof of Proposition~\ref{prop:no_gain} shows that if
$u^0(x_t;\theta)=u^0(x_t;\theta^\star)$ for $t=0,\ldots,K-1$, then
$u^0(x;\theta)=u^0(x;\theta^\star)$ for every $x$ in the subspace of the
corresponding case; nothing in that argument required $x$ to be an observed
state. Since $\theta\in\mathcal C_K$ agrees with
$\theta^\star$ at $x_0,\ldots,x_{K-1}$, the claims (i)--(iii) follow.
\end{proof}

Outside the observed affine hull the error need not vanish, but it is
controlled by the distance to the hull and by the residual ambiguity in the
gain. Let
$D:=\begin{bmatrix}x_1-x_0&\cdots&x_{K-1}-x_0\end{bmatrix}\in\R^{n\times(K-1)}$
and let $P_D:=DD^\dagger$ be the orthogonal projector onto
$\operatorname{im}(D)$, where $D^\dagger$ is the Moore--Penrose pseudoinverse.

\begin{proposition}[Off-hull generalization bound]
\label{prop:off_hull_error_bound}
Let $\theta\in\mathcal C_K$ be any parameter consistent with
the past observations. For all $x\in\R^n$, we have
$\|\rhcaction{x}{\theta}-\rhcaction{x}{\theta^\star}\|_2
 \le\|M_\theta - M_{\theta^\star}\|_2\cdot\|(I-P_D)(x-x_0)\|_2.$
\end{proposition}

\begin{proof}
By \eqref{eq:first_action_policy}, the difference of the two policies is
$(M_\theta-M_{\theta^\star})\,x+(b_\theta-b_{\theta^\star})$, and consistency
makes it zero at $x_0,\ldots,x_{K-1}$. Evaluating at $x_0$ gives
$b_\theta-b_{\theta^\star}=-(M_\theta-M_{\theta^\star})x_0$, so the difference
equals $(M_\theta-M_{\theta^\star})(x-x_0)$ for every $x$. Evaluating at
$x_t$, $t=1,\ldots,K-1$, then gives $(M_\theta-M_{\theta^\star})(x_t-x_0)=0$,
i.e., $(M_\theta-M_{\theta^\star})D=0$, and hence
$(M_\theta-M_{\theta^\star})P_D=(M_\theta-M_{\theta^\star})DD^\dagger=0$.
Therefore
$\rhcaction{x}{\theta}-\rhcaction{x}{\theta^\star}
=(M_\theta-M_{\theta^\star})(I-P_D)(x-x_0)$, and taking Euclidean norms with
submultiplicativity gives the bound.
\end{proof}
Together, the two propositions established prediction guarantee: Inside the observed affine hull every consistent parameter
reproduces the expert exactly; outside it the error grows at most linearly
with the distance to the hull.

%% file: sections/revision1_NumericalResults.tex
\section{Numerical Results}
\label{sec:numerical}
We evaluate whether an inferred objective can reliably predict actions at unseen states when observation data is noisy. We compare inverse receding-horizon LQR with two baselines. Comparison with direct policy fitting, i.e., least-squares regression of the observed actions on the observed states as in~\cite{recht2019tour}, evaluates whether retaining and re-optimizing an objective improves prediction under horizon changes. Comparison with finite-horizon inverse LQR~\cite{9029795} evaluates whether correctly modeling the receding-horizon structure of the demonstrations improves objective inference. 

\begin{figure*}[t!]
    \centering
    \includegraphics[width=1.0\textwidth]
        {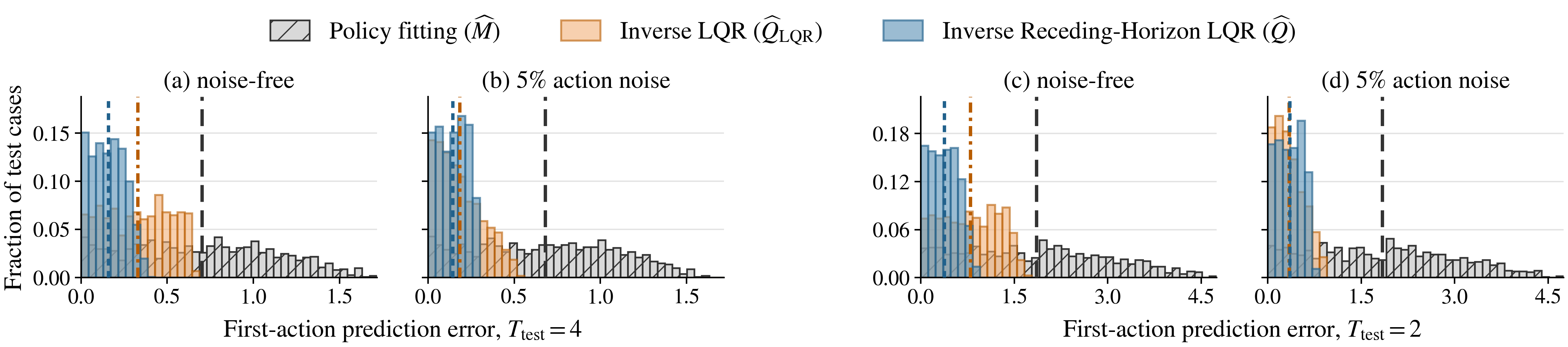}\vspace{-2em}
        \caption{\textbf{Replanning with an objective inferred by inverse receding-horizon LQR yields
low action-prediction error at unseen states under a changed horizon,
even with noisy observations.} Absolute first-action error on $1{,}000$
unseen states for policy fitting ($\widehat{M}$, gray hatched), inverse
LQR ($\widehat Q_{\mathrm{LQR}}$, orange), and inverse receding-horizon LQR
($\widehat Q$, blue), transferring from training horizon $10$ to test
horizon $4$ and $2$, with (a,c) exact observations and (b,d)
Gaussian action noise with standard deviation $5\%$ of the noiseless
action RMS. Dashed lines mark mean errors.
        }
    \label{fig:cross-horizon-generalization}
\end{figure*}
\noindent\textbf{Experimental setup.} We consider system \eqref{eq:system_dynamics} with
\begin{equation*}
    A = \begin{bmatrix}1 & 0.2\\0 & 1\end{bmatrix}\hspace{-0.1em},\hspace{-0.1em}
    B = \begin{bmatrix}0.02\\0.2\end{bmatrix}\hspace{-0.1em},\hspace{-0.1em}
    Q^\star = \begin{bmatrix}2 & 0.25\\0.25 & 0.8\end{bmatrix}\hspace{-0.1em},\hspace{-0.1em}
    R = 0.05,
\end{equation*}
where $R$ and $q=r=0$ are known and only $Q^\star$ is to be inferred. The expert uses planning horizon
$T_{\mathrm{train}}=10$, and we collect a single rollout of
$K=8$ state-action pairs from $x_0=(2,-0.5)$. For noisy observations,
each recorded action is perturbed by independent zero-mean Gaussian
noise with standard deviation equal to $5\%$ of the root mean square (RMS) noiseless action; the states and rollout remain noiseless.

All methods use the same observations. Inverse receding-horizon LQR minimizes
\eqref{eq:inverse_rhc_problem} with L-BFGS
\cite{zhang2024inverse,11312871}, using the closed-form first-action map \eqref{eq:first_action_policy} and its
gradient.
Inverse LQR treats the
$K$ observations as one open-loop LQ trajectory of horizon $N=K$,
matching each observation with the time-varying Riccati gain associated
with the remaining horizon rather than the fixed first-step gain of the
receding-horizon controller. Policy fitting estimates the
linear feedback gain by least squares,
$\widehat M = UX^\dagger$, where $X$ and $U$ stack the observed states
and actions column-wise.


\noindent\textbf{Results. }Fig.~\ref{fig:cross-horizon-generalization} shows that inverse receding-horizon LQR yields mean errors of $0.16$ and $0.38$ at $T_{\mathrm{test}}=4$ and $2$, compared with $0.70$ and $1.86$ for policy fitting, despite $\|\widehat Q-Q^\star\|_F=0.24$. Inverse LQR gives $0.33$ and $0.79$ and cannot fit the receding-horizon demonstrations exactly (training RMSE $0.07$), highlighting the effect of planning-model mismatch. With $5\%$ action noise, inverse receding-horizon LQR remains accurate ($0.14$, $0.35$) and outperforms policy fitting ($0.68$, $1.84$), while its gap to inverse LQR ($0.18$, $0.34$) narrows as noise masks the planning-model mismatch.

%% file: sections/conclusion.tex
\section{Conclusion}
An inferred objective is useful not only for explaining observed actions, but also for predicting how decisions change beyond the observed setting. Our results clarify when receding-horizon behavior provides enough information to recover objective parameters and when the recovered objective supports such predictions. The numerical comparison further illustrates the benefit of retaining the planning mechanism: an inferred cost can be reused when the planning horizon changes, whereas a directly fitted policy need not capture this dependence. Future work will extend the analysis to nonlinear dynamics and nonquadratic cost functions. 
